\documentclass[onecolumn,showpacs,notitlepage,nofootinbib,floatfix,superscriptaddress]{revtex4-1} 
\usepackage{graphicx,xcolor}
\usepackage{amssymb,amsmath,amsthm,amsfonts}
\usepackage{hyperref} 
\usepackage[ruled]{algorithm2e}

\DeclareMathOperator{\sgn}{sgn}
\newcommand{\face}[2]{{\Delta_{#1}(#2)}}
\renewcommand{\star}[2]{{\mathrm{St}_{#1}(#2)}}
\DeclareMathOperator{\im}{im}
\newcommand{\pr}[1]{\mathrm{pr}\left(#1\right)}
\newcommand{\diam}[1]{{\mathrm{diam}}(#1)}
\newcommand{\dist}[2]{d\left(#1,#2\right)}

\newcommand{\pth}[2]{(#1\mathbin{:}#2)}
\newcommand{\cpth}[2]{(#1\mathbin{\updownarrow}#2)}
\DeclareMathOperator{\usetop}{\uparrow}
\newcommand{\uset}[1]{\usetop(#1)}
\DeclareMathOperator{\dsetop}{\downarrow}
\newcommand{\dset}[1]{\dsetop(#1)}

\newcommand{\cdia}[1]{\Diamond\left(#1\right)}
\newcommand{\rest}[1]{|_{#1}}
\newcommand{\cdist}[1]{d_{\updownarrow}\left(#1\right)}

\renewcommand{\phi}{\varphi}
\renewcommand{\rho}{\varrho}
\newcommand{\tmax}{T_{\mathrm{max}}}
\newcommand{\pthr}{p^*_{\mathrm{th}}}
\newcommand{\psusval}{0.99\pm 0.02\%}
\newcommand{\gammaval}{0.855\pm 0.010}
\newcommand{\psusvaltoom}{1.98\pm 0.02\%}
\newcommand{\gammavaltoom}{0.80\pm 0.01}

\newtheorem{definition}{Definition}
\theoremstyle{definition}

\newtheorem{lemma}[definition]{Lemma}
\theoremstyle{lemma}

\newtheorem{theorem}[definition]{Theorem}
\theoremstyle{theorem}

\makeatletter
\newcommand\gobblepars{\@ifnextchar\par{\expandafter\gobblepars\@gobble}{}}
\makeatother
\newcommand{\sectionprl}[1]{{\it #1.---}\gobblepars}
\begin{document}

\title{Cellular-automaton decoders with provable thresholds for topological codes}

\author{Aleksander Kubica}
\affiliation{Perimeter Institute for Theoretical Physics, Waterloo, ON N2L 2Y5, Canada}
\affiliation{Institute for Quantum Computing, University of Waterloo, Waterloo, ON N2L 3G1, Canada}
\author{John Preskill}
\affiliation{Institute for Quantum Information \& Matter, California Institute of Technology,  Pasadena, CA 91125, USA}
\affiliation{Walter Burke Institute for Theoretical Physics, California Institute of Technology,  Pasadena, CA 91125, USA}

\date{\today}

\begin{abstract}
We propose a new cellular automaton (CA), the Sweep Rule, which generalizes Toom's rule to any locally Euclidean lattice.
We use the Sweep Rule to design a local decoder for the toric code in $d\geq 3$ dimensions, the Sweep Decoder, and rigorously establish a lower bound on its performance.
We also numerically estimate the Sweep Decoder threshold for the three-dimensional toric code on the cubic and body-centered cubic lattices for phenomenological phase-flip noise.
Our results lead to new CA decoders with provable error-correction thresholds for other topological quantum codes including the color code.
\end{abstract}

\maketitle

\twocolumngrid

%%%%%%%%%%%%%%%%%%%%%%%%%%%%%%%%%%%%%%%%%%%%
%\section{Introduction}
%%%%%%%%%%%%%%%%%%%%%%%%%%%%%%%%%%%%%%%%%%%%

To fault-tolerantly operate a scalable universal quantum computer, one protects logical information using a quantum error-correcting code, and removes errors without disturbing the encoded information \cite{Shor1996, Gottesman1996, Preskill1998}.
This can be achieved with stabilizer codes.
Each stabilizer generator is measured, yielding an outcome $\pm 1$, and a classical decoding algorithm then computes the recovery operator.
Unfortunately, optimal decoding of generic stabilizer codes is computationally hard \cite{Hsieh2011,Iyer2015}.
Thus, to render this task tractable one should restrict attention to codes with some structure.

Topological stabilizer codes \cite{Kitaev2003, Bravyi1998, Bombin2006, Haah2011,Bombin2013book,Kubica2015}, such as the toric and color codes, have a lot of structure due to the geometric locality of their stabilizer generators.
Namely, any stabilizer returning a $-1$ measurement outcome indicates the presence of errors in its neighborhood.
By exploiting this syndrome pattern, many efficient decoders with high error-correction thresholds have been proposed \cite{Dennis2002,Wang2009,Wang2011,Duclos-Cianci2010,Bravyi2013a,Fowler2012a,Anwar2013, Delfosse2014,Delfosse2017a,Delfosse2017,Nickerson2017,Maskara2018}. 
However, most of these decoders use global classical information about the measurement outcomes and thus require communication between distant parts of the system.
In any realistic setting, new faults appear during the time needed to collect and process global syndrome data~\cite{Fowler2012, Terhal2015}.
Thus, to avoid error accumulation we desire fast decoders, which ideally use only local information.

A very promising class of topological quantum code decoders is based on cellular automata (CA) \cite{Harrington2004, Herold2015, Dauphinais2017}.
CA decoders are very efficient because they naturally incorporate parallelization and can be implemented on dedicated hardware without any non-local communication.
As initially suggested in Ref.~\cite{Dennis2002}, a simple CA, called Toom's rule \cite{Toom1980, Bennett1985,Grinstein2004}, can successfully protect quantum information encoded into the 4D toric code on a hypercubic lattice.
Moreover, recent numerical simulations \cite{Pastawski2011,Breuckmann2017,Breuckmann2017a} indicate that heuristic decoders based on Toom's rule have non-zero error-correction thresholds for higher-dimensional toric codes.

In this article  we address the fundamental question whether using a CA is a viable error-correction strategy for topological quantum codes.
First, we propose a new CA, the Sweep Rule, which is a generalization of Toom's rule to any locally Euclidean lattice in $d\geq 2$ dimensions.
The Sweep Rule shrinks $(k-1)$-dimensional domain walls for any $k = 2,\ldots, d$.
Then, we use the Sweep Rule to design a new local decoder of the toric code in $d\geq 3$ dimensions, the Sweep Decoder, and rigorously prove a lower bound on its performance for perfect syndrome extraction.
Finally, we numerically demonstrate successful error suppression using a noisy version of the Sweep Rule.
In particular, we estimate the sustainable threshold error rate ${p^{\textrm{bcc}}_{\textrm{sus}} = \psusval}$ of the Sweep Decoder for phase-flip errors and imperfect syndrome measurements in the 3D toric code on the body-centered cubic (bcc) lattice; see Fig.~\ref{fig_numerics}.
Our decoding scheme works reliably against Pauli $X$ or Pauli $Z$ errors if the corresponding syndrome is at least one-dimensional and the error rate is below the theshold value;  thus it can protect topological quantum memories in $d\geq 4$ dimensions.
Our results also lead to new CA decoders for the color code in $d\geq 3$ dimensions, presented in the accompanying article~\cite{cc_decoder}.

\begin{figure}[h!]
\centering
\includegraphics[width=.9\columnwidth]{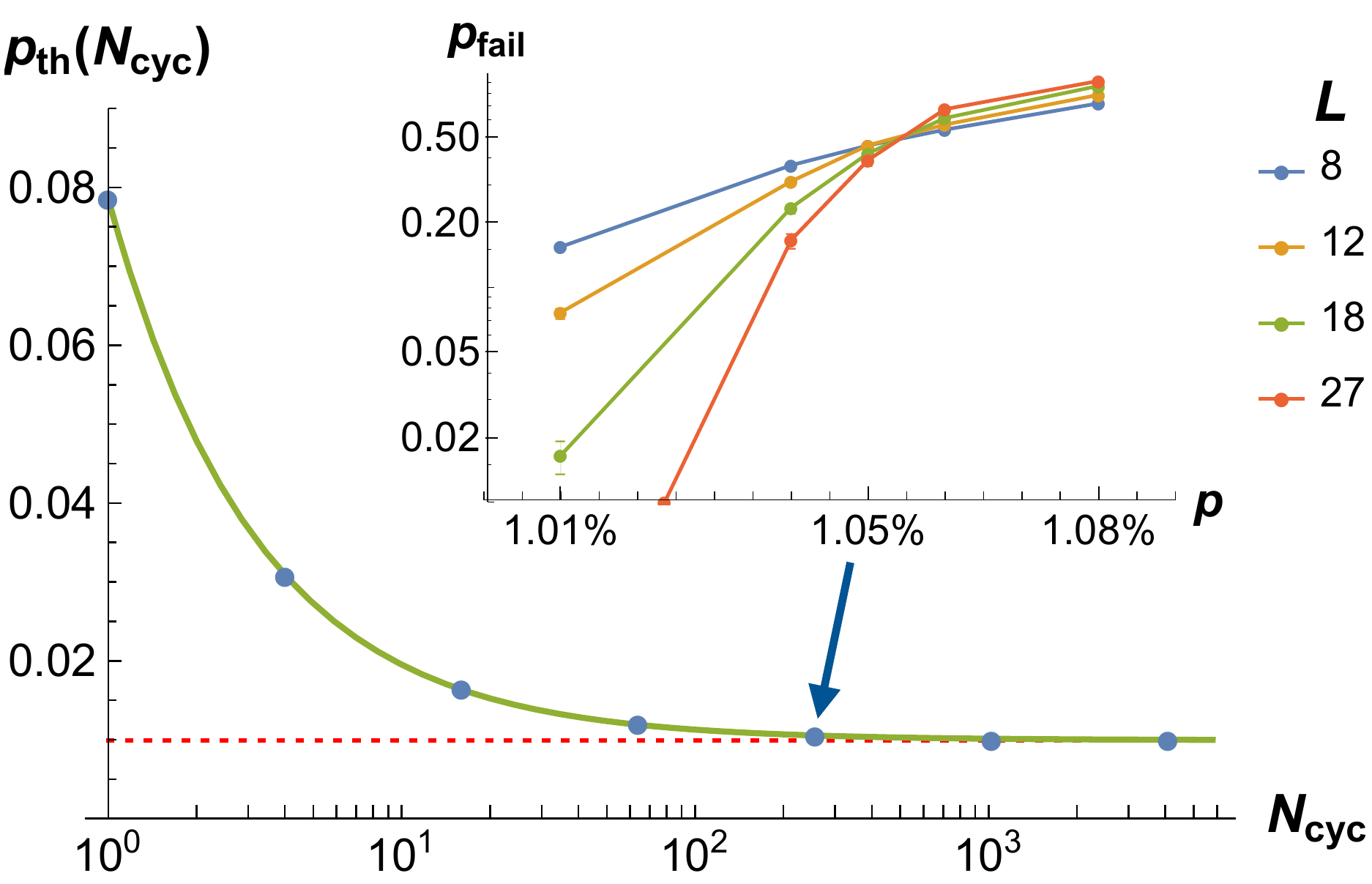}
\caption{
(Inset) The failure probability $p_{\textrm{fail}}(p,L)$ of the Sweep Decoder for the 3D toric code on the bcc lattice $\mathcal{L}$ after $N_{\textrm{cyc}} = 2^8$ correction cycles, where $p$ is the phase-flip error rate and $L$ is the linear size of $\mathcal{L}$.
We estimate the threshold $p_{\textrm{th}}(N_{\textrm{cyc}}) \approx 1.055\%$ from the crossing point of different curves.
(Main) We find the sustainable threshold $p^{\textrm{bcc}}_{\textrm{sus}} = \psusval $ by fitting the numerical ansatz from Eq.~(\ref{eq_ansatz_pheno}) to the data.
}
\label{fig_numerics}
\end{figure}

%%%%%%%%%%%%%%%%%%%%%%%%%%%%%%%%%%%%%%%%%%%%
\sectionprl{Limitations of Toom's rule}
%%%%%%%%%%%%%%%%%%%%%%%%%%%%%%%%%%%%%%%%%%%%

Consider the square lattice with a classical $\pm 1$ spin placed on every face and encode one bit of information by setting all spins to be either $+1$ or $-1$.
We want to protect the encoded bit against random spin flips, $\pm 1\mapsto \mp 1$.
This can be achieved with a CA, which flips certain spins based on \textit{locally} available information.
A simple example is the deterministic Toom's rule which sets the spin $s_C^{(T+1)}$ at time $T+1$ to 
\begin{equation}
s^{(T+1)}_C = \sgn\left(s^{(T)}_C + s^{(T)}_E + s^{(T)}_N\right),
\label{eq_toom}
\end{equation}
where $\sgn(\cdot)$ is the sign function, {$s^{(T)}_E$ and $s^{(T)}_N$ are the neighboring spins on faces to the east and north at time $T$; see Fig.~\ref{fig_Toom}(a)}
The update can be simultaneously applied to all the spins in the square lattice.

\begin{figure}[t!]
\includegraphics[width=.95\columnwidth]{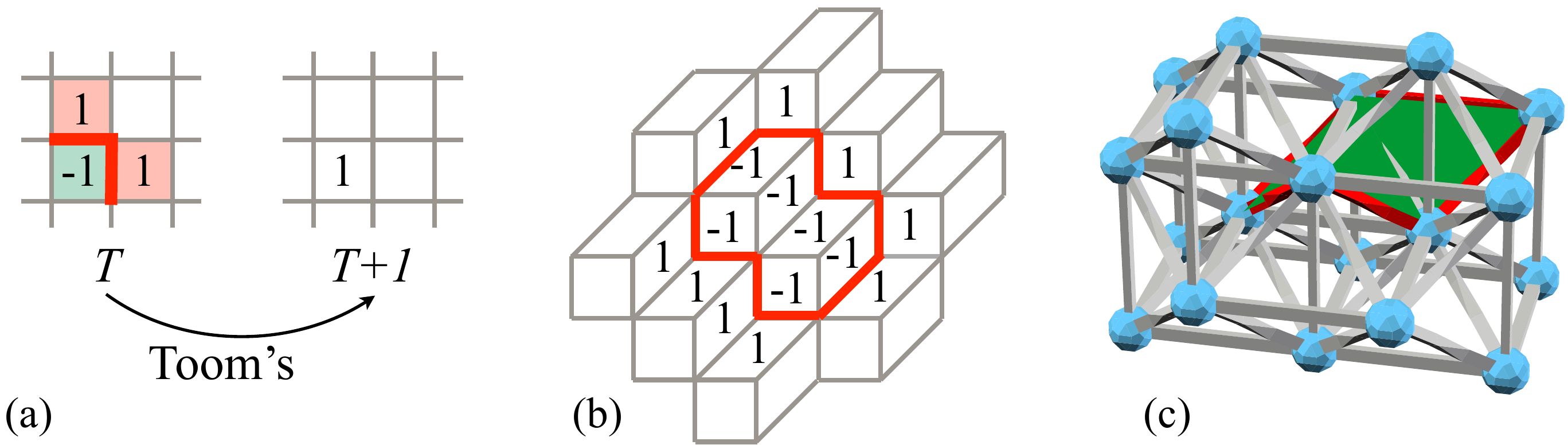}
\caption{
(a) 
At time $T$ the spin $s^{(T)}_C = -1$ (green face) differs from its neighbors to the east $s^{(T)}_E = 1$ and north $s^{(T)}_N = 1$ (red faces).{ According to Eq.~(\ref{eq_toom}), Toom's rule sets} $s^{(T+1)}_C = 1$.
(b) A 2D lattice built of three types of parallelograms with a domain wall (red), which cannot be removed by repeated application of a naive generalization of Toom's rule.
(c) The 3D toric code on the bcc lattice \cite{vzome} has qubits on faces and $X$-stabilizers associated with edges.
Any configuration of $Z$ errors (green) results in a 1D loop-like $X$-syndrome (red).
}
\label{fig_Toom} 
\end{figure}

We can rephrase Toom's rule as a conditional spin update determined by the local configuration of the 1D domain wall, i.e., the set of all edges of the lattice separating faces with spins of different value.
Let $\epsilon^{(T)}$ and $\sigma^{(T)}$ denote the set of faces with $-1$ spins and the corresponding domain wall at time $T=1,2,\ldots$.
We write $\sigma^{(T)} = \partial_2 \epsilon^{(T)}$ to capture the fact that $\sigma^{(T)}$ is the boundary of $\epsilon^{(T)}$ containing all the edges bounding faces in $\epsilon^{(T)}$.
Then, Toom's rule flips a spin on some face $f$, i.e., $s^{(T+1)}_f = - s^{(T)}_f$, iff the east and north edges of $f$ belong to $\sigma^{(T)}$; see Fig.~\ref{fig_Toom}(a).
If we know $\sigma^{(T)}$ and the set of all spins flipped between time  $T$ and $T+1$, which we denote by $\rho^{(T)}$, then the domain wall at time $T+1$ is 
\begin{equation}
\sigma^{(T+1)}  = \sigma^{(T)} + \partial_2 \rho^{(T)}.
\end{equation}
Note that this update does not require the knowledge of the actual spin values but only the locations of flipped spins, and from that perspective it may be viewed as a local rule governing the dynamics of the domain wall.
Moreover, if the domain wall disappears by time $T$,
i.e., $\sigma^{(T)} = 0$, then $\rho = \sum_{i=1}^{T-1} \rho^{(i)}$ can be viewed as an estimate
\footnote{This strategy, however, is neither guaranteed to terminate nor to return $\rho = \epsilon^{(1)}$.}
of $\epsilon^{(1)}$ with the boundary $\partial_2 \rho$ matching the initial domain wall $\sigma^{(1)}$.
As we will see later, correcting errors in the toric code in $d \geq 3$ dimensions can also be rephrased as estimating $\epsilon^{(1)}$ given its boundary $\sigma^{(1)}$, by exploiting the domain-wall structure of the  syndrome.

This version of  Toom's rule works for the square lattice, but it is not obvious how to generalize it to other 2D lattices, or to higher dimensions.
To illustrate the difficulty, consider the 2D lattice  in Fig.~\ref{fig_Toom}(b).
If one uses a simple update rule \textit{``flip a spin iff east and north edges of the face belong to the domain wall''}, then there exist spin configurations with domain walls which cannot be removed by repeated application of this rule.
For such error syndromes, the Toom's rule decoder fails to correct the erroneous spins.
To define a workable version of Toom's rule, the lattice must have suitable properties, which we now specify.

%%%%%%%%%%%%%%%%%%%%%%%%%%%%%%%%%%%%%%%%%%%%
\sectionprl{Causal lattices}
%%%%%%%%%%%%%%%%%%%%%%%%%%%%%%%%%%%%%%%%%%%%

We consider a lattice $\mathcal{L}$, which is a triangulation (possibly without any symmetries) of the Euclidean space $\mathbb{R}^2$.
We denote by $\face i {\mathcal{L}}$ the set of all $i$-simplices of $\mathcal{L}$.
In particular, $\face 0 {\mathcal{L}}$, $\face 1 {\mathcal{L}}$ and $\face 2 {\mathcal{L}}$ correspond to vertices, edges and triangular faces of $\mathcal{L}$.
We assume that each $\face i {\mathcal{L}}$ contains countably many elements and define the sweep direction as a unit vector $\vec t \in \mathbb{R}^2$ not perpendicular to any edge of $\mathcal{L}$.

We define a path $\pth u w$ between two vertices $u$ and $w$ of the lattice $\mathcal{L}$ to be a collection of edges
{$(u,v_1),\ldots,(v_{n},w) \in\face{1}{\mathcal{L}}$, where $v_i\in\face{0}{\mathcal{L}}$.}
If the sign of the inner product $\vec{t}\cdot (v_i,v_{i+1})$ is the same for all edges in the path $\pth u w$, then we call the path causal and denote it by $\cpth u w$.
We remark that any pair of the vertices of $\mathcal{L}$ is connected by a path but there might not exist a causal path between them; see Fig.~\ref{fig_causal_structure}(a).
Finally, we define the causal distance
\begin{equation}
\label{eq_cdist}
\cdist{u,w} = \min_{\cpth u w} |\cpth u w |
\end{equation}
to be the length of the shortest causal path between $u$ and $w$; if there is no causal path, then $\cdist{u,w} = \infty$.

We observe that the sweep direction $\vec t$ induces a binary relation $\preceq$ over the set of vertices $\face{0}{\mathcal{L}}$.
We say that $u$ precedes $w$, i.e., $u\preceq w$ for $u,v\in \face 0 {\mathcal{L}}$, iff there exists a causal path $\cpth u w$ and  $\vec{t} \cdot (v_i,v_{i+1}) > 0 $ for any edge $(v_i,v_{i+1})\in \cpth u w$.
Equivalently, we write $w\succeq u$ and say that $w$ succeeds $u$.
Abusing the notation, we write $v \preceq \kappa$ if all vertices $\face 0 \kappa$ of a $k$-simplex $\kappa\in\face k {\mathcal{L}}$ succeed $v$, i.e., $v\preceq u$ for all $u\in\face 0 \kappa$; similarly for $\kappa \preceq v$.

\begin{figure}[t!]
\centering
\includegraphics[width = .95\columnwidth]{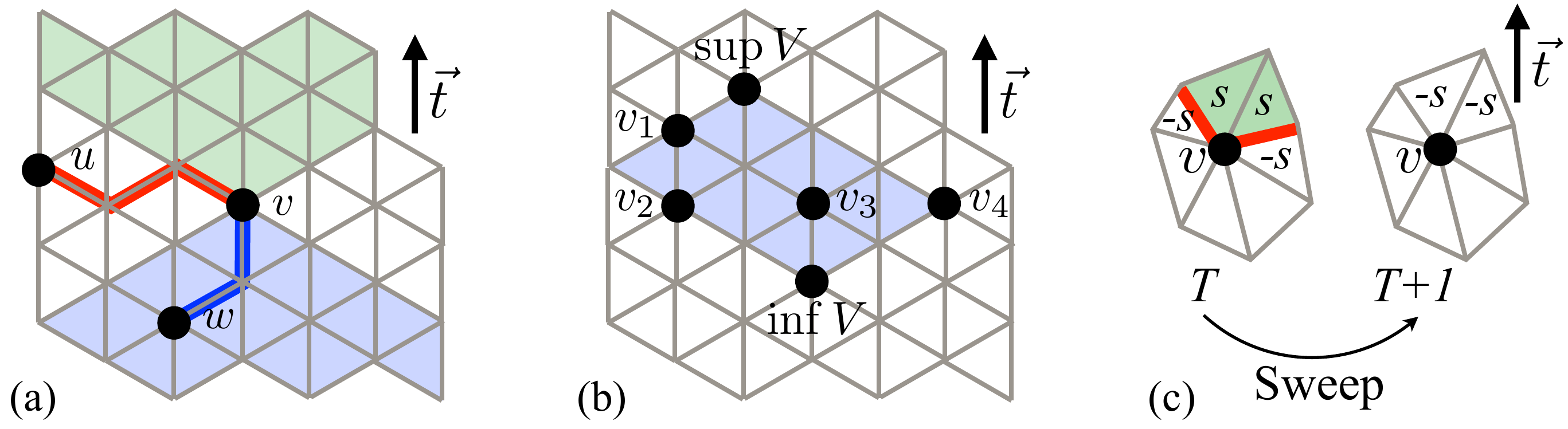}
\caption{
(a) {Vertices} $u$ and $v$ are connected by a path $\pth u v$ (red), but there is no causal path between them; $v$ and $w$ are connected by a causal path $\cpth u v $ (blue).
We shaded in green and blue the future $\uset v$ and past $\dset v$ of $v$.
(b) The causal diamond $\cdia V$ (blue) of a subset of vertices $V = \{ v_1,v_2,v_3,v_4\}$ is defined as the intersection of the future of the infimum of $V$ with the past of the supremum of $V$.
(c) The Sweep Rule is defined for every vertex and locally updates $\pm 1$ spins on neighboring faces.
Since the vertex $v$ is trailing, spins on two green faces will be flipped.
}
\label{fig_causal_structure} 
\end{figure}

We can think of the partial order $\preceq$ between vertices of the lattice as a causality relation between points in the discretized $(1+1)$D spacetime with $\vec t$ corresponding to the time
\footnote{We warn the reader that later we use the time $T$ to index how many times the CA rule is applied.}
direction; see Fig.~\ref{fig_causal_structure}(a)(b).
We define the future $\uset v$ and past $\dset v$ of a vertex $v\in\face{0}{\mathcal{L}}$ as the collection of all simplices of $\mathcal{L}$ succeeding and preceding $v$, namely
\begin{eqnarray}
\uset v &=& \{ \kappa\in\face{k}{\mathcal{L}} | \forall k\textrm{ and } v\preceq \kappa \},\\
\dset v &=& \{ \kappa\in\face{k}{\mathcal{L}} | \forall k\textrm{ and } v\succeq \kappa \}.
\end{eqnarray}
Every finite subset of vertices $V\subseteq\face 0 {\mathcal{L}}$ has a unique supremum, the vertex $\sup V$, where $\sup V$ lies in the future of each $u \in V$, and furthermore $\sup V$ lies in the past of each vertex $w$ which is in the future of each $v\in V$.
The infinum $\inf V$ is defined analogously.
Lastly, we define the causal diamond $\cdia V$ as the intersection of the future of $\inf V$ and the past of $\sup V$, i.e.,
\begin{equation}
\cdia V = \uset{\inf V} {\cap} \dset{\sup V}.
\end{equation}

This discussion of causal structure generalizes to lattices embedded in a torus; however, one has to excercise caution since the partial order is well-defined only within local regions.
Also, in case of higher-dimensional lattices we make certain assumptions about their causal structure, such as the existence of unique infimum and supremum of $V$.
To avoid technical details, we simply refer to lattices satisfying those assumptions as causal and defer the discussion to Appendix~\ref{app_lattices}.

%%%%%%%%%%%%%%%%%%%%%%%%%%%%%%%%%%%%%%%%%%%%
\sectionprl{Sweep Rule}
%%%%%%%%%%%%%%%%%%%%%%%%%%%%%%%%%%%%%%%%%%%%

Let $\mathcal{L}$ be a 2D causal lattice with $\pm 1$ spins on triangular faces and $\epsilon\subseteq\face{2}{\mathcal{L}}$ denote the set of all faces with $-1$ spins.
The corresponding domain wall $\sigma$ can be found as the boundary $\partial_2 \epsilon$.
Let $v$ be a vertex of $\mathcal{L}$ and denote by $\sigma \rest v$ the restriction of the domain wall $\sigma$ to the edges incident to $v$.
We say that $v$ is trailing if $\sigma \rest v$ is non-empty and belongs to the future of $v$, namely $\sigma\rest v \subset\  \uset v$; see Fig.~\ref{fig_Sweep_illustrated}.
We propose a new local spin update rule defined for every vertex $v$ of $\mathcal{L}$.
\begin{definition}[Sweep Rule]
\label{def_sweep2}
If a vertex $v$ is trailing, then find a subset of neighboring faces $\phi(v)$ in the future $\uset v$ with boundary locally matching the domain wall, i.e., $(\partial_2 \phi(v))\rest v = \sigma\rest v$, and flip spins on faces in $\phi(v)$.
\end{definition}
\noindent This Rule is deterministic and there is a unique choice of $\phi(v)$.
The spin update results in the domain wall being locally pushed away from any trailing vertex $v$; see Fig.~\ref{fig_Sweep_illustrated}.
Note that nothing happens if a vertex is not trailing.
We can, however, consider a very similar CA, the Greedy Sweep Rule, which always tries to push the domain wall away from $v$ in the sweep direction $\vec t$, irrespective of $v$ being trailing; see Appendix~\ref{app_sweep}.

In Lemma~\ref{lemma_properties} we present properties of the Sweep Rule (proven in Appendix~\ref{app_properties}) needed to establish a non-zero threshold of the Sweep Decoder.

 \begin{figure}[t!]
\centering
\includegraphics[width=.95\columnwidth]{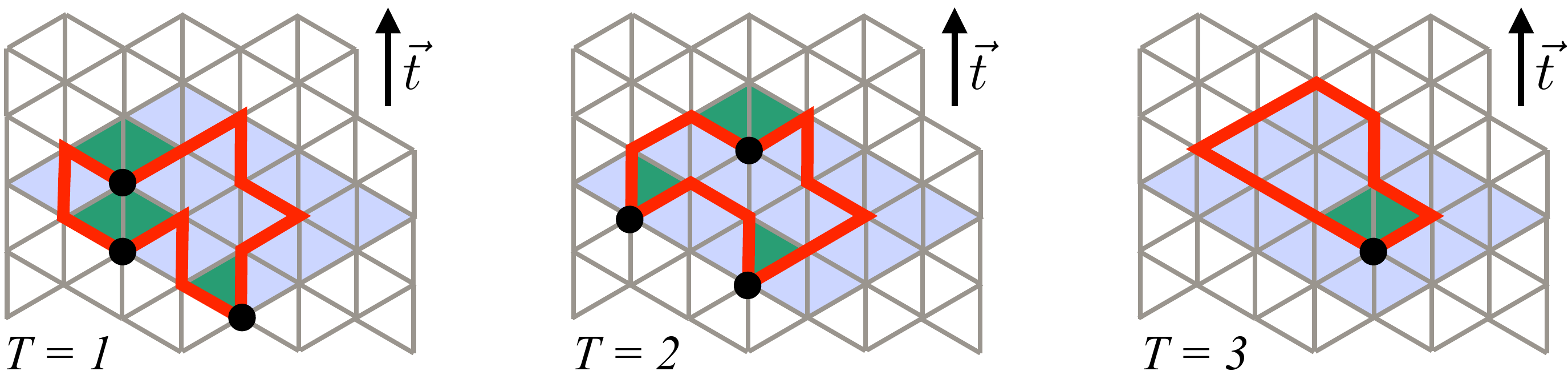}
\caption{
{For each trailing vertex $v$ (black) at time $T=1,2,3$ the Sweep Rule finds} a subset $\phi(v)$ of neighbouring faces (green) in the future $\uset v$, whose boundary $\partial_2 \phi(v)$ locally matches the domain wall $\sigma^{(T)}$ (red), i.e., $(\partial_2 \phi(v))\rest v = \sigma^{(T)}\rest v$.
Flipping spins in $\phi(v)$ pushes $\sigma^{(T)}$ away from $v$ in the sweep direction $\vec t$.
Note that $\phi(v)$ and $\sigma^{(T)}$ are always in the causal diamond $\cdia{\sigma^{(1)}}$ (blue) of the initial domain wall $\sigma^{(1)}$.
}
\label{fig_Sweep_illustrated}
\end{figure}

\begin{lemma}[Sweep Rule Properties]
\label{lemma_properties}
Let $\sigma$ be a domain wall in the causal lattice $\mathcal{L}$.
If the Sweep Rule is simultaneously applied to every vertex of $\mathcal{L}$ at time steps $T=1,2,\ldots$, then
\begin{enumerate}
\item (Support) the domain wall $\sigma^{(T)}$ at time $T$ stays within the causal diamond $\cdia \sigma$, i.e., 
\begin{equation}
\sigma^{(T)} \subset \cdia \sigma,
\label{eq_support}
\end{equation}
\item (Propagation) the causal distance between $\sigma$ and any vertex $v$ of $\sigma^{(T)}$ is at most $T$, i.e., 
\begin{equation}
\cdist{v,\sigma}  \leq T,
\end{equation}
\item (Removal) 
the domain wall is removed by time $T^*$, i.e., $\sigma^{(T)} = 0$ for all $T> T^*$, where
\begin{equation}
T^* = \max_{\cpth{\inf\sigma}{\sup\sigma}}  |\cpth{\inf\sigma}{\sup\sigma}|.
\label{eq_removal}
\end{equation}
\end{enumerate}
\end{lemma}

The Sweep Rule may also be defined for vertices of a $d$-dimensional causal lattice $\mathcal{L}$ with spins placed on $k$-simplices $\face k {\mathcal{L}}$, where $k =2, \ldots, d$.
However, for $k \neq d$ the local choice of spins to flip $\phi(v)$ may not be unique (this does not happen in 2D).
Thus, we consider a family of rules corresponding to different ways of choosing $\phi(v)$ in such a way that, roughly speaking, the local causal structure of the domain wall is preserved after flipping spins on $k$-simplices in $\phi(v)$; see Appendix~\ref{app_sweep}.

%%%%%%%%%%%%%%%%%%%%%%%%%%%%%%%%%%%%%%%%%%%%
\sectionprl{Sweep Decoder}
%%%%%%%%%%%%%%%%%%%%%%%%%%%%%%%%%%%%%%%%%%%%

We may use the $d$-dimensional version of the Sweep Rule to decode the toric code on the $d$-dimensional causal lattice $\mathcal{L}$.
Recall that the toric code of type $k=1,\ldots,d-1$ {is} defined by placing qubits on $k$-simplices of $\mathcal{L}$, and associating $X$- and $Z$-stabilizers with $(k-1)$- and $(k+1)$-simplices.
Then, $Z$-stabilizers, $Z$-logical operators and $X$-syndromes correspond to, respectively, the elements of $\im\partial_{k+1}$, $\ker\partial_k$ and $\im\partial_k$, where $\partial_i$ denotes the $i$-boundary operator; see Appendix~\ref{app_lattices}.
If $\epsilon\subseteq\face k {\mathcal{L}}$ is the set of qubits affected by  $Z$ errors, then the corresponding $X$-syndrome is $\sigma = \partial_k \epsilon$.
Thus, for $k \geq 2$, decoding of $Z$ errors can be phrased as the already discussed problem of estimating locations of $-1$ spins given the corresponding domain wall.
{Note} that for $k \leq d-2$ decoding of $X$ errors is analogous but in the dual lattice $\mathcal{L}^*$ with the $Z$-syndrome forming a $(d-k-1)$-dimensional domain wall.

\begin{algorithm}%[H]
\caption{Sweep Decoder}
\SetKwInOut{Require}{Require}
\KwIn{$X$-syndrome $\sigma\in\textrm{im }\partial_{k}$, {$k=2,\ldots,d-1$}}
\KwOut{$k$-dimensional correction $\rho\subseteq\face{k}{\mathcal{L}}$}
initialize $T$ = 1, $\sigma^{(1)} = \sigma$\\
unless $T> \tmax $ or $\sigma^{(T)} = 0$ repeat:\\
\begin{enumerate}
\item apply the Sweep Rule simultaneously to every\\ vertex of $\mathcal{L}$ to get $\rho^{(T)}$
\item find $\sigma^{(T+1)} = \sigma^{(T)} + \partial_k\rho^{(T)}$
\item update time step $T\leftarrow T+1$\\
\end{enumerate}
if $T\leq \tmax$
\footnote{$\tmax$ is of the order of the linear size $L$ of the lattice $\mathcal{L}$; see Appendix~\ref{app_proof_together}.}, then $\rho = \sum_{i=1}^{T-1}\rho^{(i)}$, otherwise $\rho = \texttt{FAIL}$\\
\KwRet{$\rho$}
\end{algorithm}

This \textit{Sweep Decoder} may fail for either one of two reasons.
First, it might not terminate within time $\tmax$,
which results in $\rho = \texttt{FAIL}$.
Second, the correction $\rho$ combined with the initial error $\epsilon$ may implement a non-trivial logical operator, i.e., $\rho + \epsilon \not\in \im \partial_{k+1}$.
{However}, the Sweep Decoder a has non-zero error-correction threshold --- if the $Z$ error rate is below threshold, then the failure probability rapidly approaches zero as the code's block grows.
We establish this fact by deriving a lower bound $\pthr > 0$ on the threshold error rate (which explicitly depends on the local structure of $\mathcal{L}$).
\begin{theorem}[Threshold]
\label{thm_thres}
Consider a family of causal lattices $\mathcal{L}$ of growing linear size $L$ on the $d$-dimensional torus, and define the toric code of type $k=2,\ldots,d-1$ on $\mathcal{L}$.
Then, there exists a constant $\pthr > 0$, such that for any phase-flip error rate $p < \pthr$ the failure probability of the Sweep Decoder for perfect syndrome extraction goes to zero as $L\rightarrow\infty$.
\end{theorem}
\noindent In Appendix~\ref{app_proof} we present a rigorous proof of Theorem~\ref{thm_thres} based on renormalization group ideas \cite{Gacs1988, Harrington2004, Bravyi2013a}; here we  only outline the proof strategy.

\begin{proof}
First, we decompose each error configuration into recursively defined ``connected components,'' where a ``level-$n$'' connected component has a linear size growing exponentially with $n$.
The probability of a level-$n$ connected component is doubly-exponentially small in $p / \pthr$.
The connected components are well isolated from other errors; therefore, using  Lemma~\ref{lemma_properties} and some modest assumptions about the lattice family, we can show that a connected component with linear size small compared to $L$ will be successfully removed by repeated application of the Sweep Rule.
Therefore, the Sweep Decoder fails only if the contains a level-$n$ connected component with size comparable to $L$, which is very improbable for large $L$ and $p < \pthr$.
\end{proof}

%%%%%%%%%%%%%%%%%%%%%%%%%%%%%%%%%%%%%%%%%%%%
\sectionprl{Numerical simulations}
%%%%%%%%%%%%%%%%%%%%%%%%%%%%%%%%%%%%%%%%%%%%

In Theorem~\ref{thm_thres} we assumed that the Sweep Rule is applied flawlessly, but in a realistic scenario the Rule itself is noisy; the noise degrades the effectiveness of error correction and reduces the threshold. 
We have numerically investigated the performance of the Sweep Decoder for the 3D toric code on the bcc lattice with qubits on faces. 
We consider a phenomenological noise model such that in each error correction cycle Pauli $Z$ errors on qubits occur with probability $p$, and in addition measured syndrome bits are flipped with probability $p$.
Using Monte Carlo simulations we find the threshold $p_{\textrm{th}}(N_{\textrm{cyc}})$ for a fixed number $N_{\textrm{cyc}}$ of noisy correction cycles followed by perfect syndrome extraction and full decoding.
Note that $p_{\textrm{th}}(1)$ is the threshold for perfect syndrome extraction.
We are, however, interested in the so-called sustainable threshold $p^{\textrm{bcc}}_{\textrm{sus}} = \lim_{N_{\textrm{cyc}}\rightarrow\infty} p_{\textrm{th}}(N_{\textrm{cyc}})$ \cite{Brown2015, Terhal2015}.
We observe that the threshold $p_{\textrm{th}}(N_{\textrm{cyc}})$ is very well approximated by the numerical ansatz
\begin{equation}
p_{\textrm{th}}(N_{\textrm{cyc}}) \sim p^{\textrm{bcc}}_{\textrm{sus}} (1 - (1- p_{\textrm{th}}(1)/p^{\textrm{bcc}}_{\textrm{sus}}) N_{\textrm{cyc}}^{-\gamma}),
\label{eq_ansatz_pheno}
\end{equation}
with the fitting parameters ${p^{\textrm{bcc}}_{\textrm{sus}} = \psusval}$ and ${\gamma = \gammaval}$; see~Fig.~\ref{fig_numerics}.
These numerical results were actually obtained for a variant of the Sweep Decoder based on the Greedy Sweep Rule, which has a higher threshold than the decoder based on the Sweep Rule.
In Appendix~\ref{app_sweep} we discuss the Greedy Sweep Rule, explain how it generalizes to locally Euclidean lattices, and use it to estimate the sustainable threshold of the 3D toric code on the cubic lattice $p^{\textrm{cubic}}_{\textrm{sus}} = \psusvaltoom $.

%%%%%%%%%%%%%%%%%%%%%%%%%%%%%%%%%%%%%%%%%%%%
\sectionprl{Discussion}
%%%%%%%%%%%%%%%%%%%%%%%%%%%%%%%%%%%%%%%%%%%%

We have presented a new CA, the Sweep Rule, which generalizes Toom's rule to any locally Euclidean $d$-dimensional lattice.
This Rule can be used to decode a topological quantum code whose error syndrome is at least one dimensional, including the color code; see the accompanying article~\cite{cc_decoder}.
We proved that a decoder based on the Sweep Rule has a non-zero accuracy threshold for the toric code, and we numerically studied its performance against a phenomenological noise model.

Our results provide a rigorous justification for using CA error-correction strategies for topological quantum codes.
We hope that our proof techniques will lead to new CA decoders with provable thresholds for codes on lattices with boundaries, hyperbolic lattices or other quantum low-density parity-check codes.

The Sweep Rule may also be of independent interest for defining statistical-mechanical problems inspired by quantum information \cite{Kubica2014,Yoshida2014,Kubica2017}.
As for Toom's rule, one can consider a non-deterministic variant of the Sweep Rule and study the evolution of spins generated by this probabilistic CA.
We conjecture that the resulting spin dynamics is non-ergodic and that the phase diagram contains regions with multiple coexisting stable phases, as established in 2D by Toom~\cite{Toom1980}.

%%%%%%%%%%%%%%%%%%%%%%%%%%%%%%%%%%%%%%%%%%%%
\sectionprl{Acknowledgements}
%%%%%%%%%%%%%%%%%%%%%%%%%%%%%%%%%%%%%%%%%%%%

A.K. thanks Nicolas Delfosse for invaluable feedback throughout the project, and Ben Brown and Mike Vasmer for stimulating discussions.
A.K. acknowledges funding provided by the Simons Foundation through the ``It from Qubit'' Collaboration.
Research at Perimeter Institute is supported by the Government of Canada through Industry Canada and by the Province of Ontario through the Ministry of Research and Innovation.
J.P. acknowledges support from ARO, DOE, IARPA, NSF, and the Simons Foundation. The Institute for Quantum Information and Matter (IQIM) is an NSF Physics Frontiers Center.

\appendix
\onecolumngrid

%%%%%%%%%%%%%%%%%%%%%%%%%%%%%%%%%%%%%%%%%%%%
\section{Causal lattices}
%%%%%%%%%%%%%%%%%%%%%%%%%%%%%%%%%%%%%%%%%%%%
\label{app_lattices}

Let us revisit some ideas and definitions related to lattices, which we did not elaborate on in the main text but will be necessary in the proofs of the (Sweep Rule Properties) Lemma~\ref{lemma_properties} and the (Threshold) Theorem~\ref{thm_thres}.
A $d$-dimensional lattice $\mathcal{L}$ can be constructed by attaching $d$-dimensional cells to one another along their $(d-1)$-dimensional faces; see \cite{Glaser1972, Hatcher2002}.
We are particularly interested in cases when the lattice $\mathcal{L}$ is built of simplices, i.e., for any $k$ all of the $k$-cells of $\mathcal{L}$ are just $k$-simplices.
We denote by $\face k {\mathcal{L}}$ the set of all $k$-simplices of the lattice $\mathcal{L}$, where $k=0,1,\ldots, d$.

In order to describe the Sweep Rule in $d \geq 3$ dimensions, we need to discuss $k$-dimensional domain walls, where $k = 1, \ldots, d-1$.
First, we define $C_k$ to be an $\mathbb{F}_2$-vector space with the set $\face k {\mathcal{L}}$ as a basis.
Note that there is a one-to-one mapping between vectors in $C_k$ and subsets of $\face k {\mathcal{L}}$.
Then, we introduce a $(k-1)$-boundary operator $\partial_k:C_k \rightarrow C_{k-1}$ as a linear map specified for every basis element $\kappa\in\face k {\mathcal{L}}$ by
\begin{equation}
\partial_k \kappa= \sum_{\nu\in\face {k-1} \kappa} \nu,
\end{equation}
where $\face {k-1} \kappa$ is the set of all $(k-1)$-simplices contained in $\kappa$.
Let us place a $\pm 1$ spin on every $k$-simplex of $\mathcal{L}$ and denote by $\epsilon \subset \face k {\mathcal{L}}$ the locations of $-1$ spins.
The corresponding domain wall $\sigma$ can be found as the $(k-1)$-boundary of $\epsilon$, i.e., $\sigma = \partial_k \epsilon$. 
Note that by definition $(k-1)$-dimensional domain walls and elements of $\im\partial_k$ are equivalent.

In our proofs, we use two notions of distance.
We have already introduced the causal distance as the length of the shortest causal path; see Eq.~(\ref{eq_cdist}).
The other quantity, the distance $d(u,v)$ between two vertices $u$ and $v$, is defined to be the length of the shortest path connecting $u$ and $v$ in the lattice $\mathcal{L}$, namely
\begin{equation}
d(u,v) = \min_{\pth u v} |\pth u v|.
\end{equation}
We note that there is always a path between $u$ and $v$, but a causal path might not exist.
Moreover, the following inequality between the distance and the causal distance holds:
\begin{equation}
d(u,v) \leq \cdist{u,v}.
\end{equation}
We define the distance $d(U,V)$ between two subsets of vertices $U$ and $V$ as the minimal distance between any two vertices of $U$ and $V$, namely
\begin{equation}
d(U,V) = \min_{u\in U, v\in V} d(u,v).
\end{equation}
We also introduce the diameter of a subset of vertices $V$ as the maximal distance between any two vertices of $V$, i.e.,
\begin{equation}
\diam V = \max_{u,v \in V} d(u,v).
\end{equation}
Finally, we remark that the above definitions (as well as the partial order $\preceq$ between vertices of $\mathcal{L}$ induced by the sweep direction $\vec t$) are unambiguous if $\mathcal{L}$ is a discretization of the Euclidean space $\mathbb{R}^d$ and $\vec t \in \mathbb{R}^d$ is chosen not to be perpendicular to any edge of $\mathcal{L}$.
However, in the case when $\mathcal{L}$ is defined on the $d$-dimensional torus, one has to exercise caution.
For instance, the partial order can be consistently defined only within a local region of $\mathcal{L}$ with diameter at most some fraction of the linear size of $\mathcal{L}$.

To succinctly describe the Sweep Rule, we introduce a couple of notions capturing the local structure of the lattice $\mathcal{L}$.
Let $\kappa\in \face k {\mathcal{L}}$ be a $k$-simplex.
We denote by $\face l \kappa$ the set of all $l$-simplices contained in $\kappa$, where $l\leq k$.
Also, we denote by $\star n \kappa$ the set of all $n$-simplices in the neighborhood of $\kappa$ which contain $\kappa$ (this is also known as the $n$-star of $\kappa$, where $n\geq k$).
Lastly, we define a discrete $d$-dimensional ball $B_v(r)$ of radius $r$ centered at the vertex $v$ to be a collection of all $k$-simplices for any $k=0,\ldots,d$, whose distance from $v$ is less than $r$, namely
\begin{equation}
B_v(r) =\{ \kappa\in\face{k}{\mathcal{L}} | \forall k \textrm{ and } d(v,\kappa) < r \}.
\label{eq_ball}
\end{equation}
Note that a unit ball $B_v(1)$ corresponds to the collection of all simplices containing $v$, i.e., $B_v(1) = \bigsqcup_{k=0}^d \star k v$.
Also, if $\sigma$ is some collection of simplices, then the restriction of $\sigma$ to the neighborhood of $v$ is defined as $\sigma\rest v = \sigma \cap B_v(1)$.

Now we discuss necessary assumptions on the $d$-dimensional lattices to unambiguously define the Sweep Rule and prove a non-zero threshold of the Sweep Decoder.
We say that a family of lattices $\mathcal{L}$ of growing linear size $L$ is causal if it satisfies the following properties.
\begin{itemize}
\item Causal structure:
\begin{itemize}
\item[(i)] for any subset of vertices $V\subset \face 0 {\mathcal{L}}$ within a local region of $\mathcal{L}$ there exists a unique causal diamond $\cdia V$,
\item[(ii)] for any $v\in\face 0 {\mathcal{L}}$ and $\sigma \in \im \partial_k$ if $\sigma\rest v \subset\ \uset v$, there exists $\varphi(v) \subseteq \star{k} v \cap \uset v$ satisfying $(\partial_{k}\varphi(v))\rest v = \sigma\rest v$ and $\cdia{\varphi(v)} = \cdia{\sigma\rest v}$.
\end{itemize}
\item Locally Euclidean:
\begin{itemize}
\item[(iii)] for any ball $B_v(R)$ of radius $R$ within a local region of $\mathcal{L}$ one finds a cover 
\begin{equation}
\bigcup_{u\in U} B_u(r) \supset B_v(R)
\end{equation}
with balls of radius $r< R$ indexed by $U\subset \face 0 {\mathcal{L}}$, such that
\begin{equation}
|U| \leq c_B (R/r)^d
\label{eq_balls}
\end{equation}
with $d$ being the dimension of the lattice $\mathcal{L}$ and $c_B$ is a constant,
\item[(iv)] for any subset of vertices $V\subset \face 0 {\mathcal{L}}$ within a local region of $\mathcal{L}$ the diameters of $V$ and the causal diamond $\cdia V$ are comparable, i.e., there exists a constant $c_D$ such that 
\begin{equation}
\diam{\cdia V} \leq c_D \cdot \diam V,
\label{eq_diameters}
\end{equation}
\item[(v)] for any pair of vertices $u\preceq v$ the distance between them and the maximal length of any causal path between them are comparable, i.e., there exists a constant $c_P$ such that 
\begin{equation}
\max_{\cpth{u}{v}} |\cpth{u}{v}| \leq c_P \cdot \dist u v .
\label{eq_paths}
\end{equation}
\end{itemize}
\end{itemize}
Note that we require that the constants $c_B$, $c_D$ and $c_P$ do not depend on $L$.
Also, $c_D \geq 1$ since $\cdia V \supset V$.
Moreover, $c_P \geq d$ since by choosing $u = \inf \delta$ and $v = \sup \delta$ for any $d$-simplex $\delta$ we get $\dist u v  = 1$ and $\max_{\cpth u v } |\cpth u v | \geq d$.
We remark that condition (i) states that the partially ordered set $(\face 0 {\mathcal{L}},\preceq)$ is a (locally complete) lattice in the sense of order theory~\cite{Davey2002}. However, we refrain from using this term in order to avoid confusion with the lattice $\mathcal{L}$ defined as a collection of simplices.

We emphasize that the properties regarding the causal structure are sufficient if one wants to define the Sweep Rule on $\mathcal{L}$. 
Additionally, one requires the property of being locally Euclidean in order to prove that the Sweep Decoder has non-zero threshold for the toric code of type $k = 2,\ldots,d-1$ defined on $\mathcal{L}$.
Note that hyperbolic lattices do not satisfy the property of being locally Euclidean, and thus we cannot readily establish lower-bounds on the performance of the Sweep Decoder in that setting.

The aforementioned assumptions can be easily checked for translationally-invariant lattices, such as the 3D bcc lattice used to study the threshold of the Sweep Decoder for the toric code.
Indeed, let us identify the set of vertices of the 3D bcc lattice with the elements in $(2\mathbb{Z})^3 \cup (2\mathbb{Z}+1)^3$ and choose the sweep direction to be $\vec{t} = (1,1,1)\in\mathbb{R}^3$.
Then, one can explicitly find a unique infimum and supremum for any pair of vertices $u, v\in\face 0 {\mathcal{L}}$.
By induction, one can prove the uniqueness of the infimum and supremum for any finite subset of vertices, which in turn implies condition (i).
One can verify that condition (ii) is satisfied by exhaustively checking it for every possible choice of $\sigma\rest v$.
Since the bcc lattice is locally Euclidean, conditions (iii)--(v) hold straightforwardly.
We remark that it may be challenging to verify conditions (i) and (ii) for less regular lattices, for instance one may have to independently check condition (ii) for every vertex of $\mathcal{L}$.
However, we conjecture that these conditions are satisfied by any lattice built of simplices which is a discretization of the Euclidean space $\mathbb{R}^d$.

%%%%%%%%%%%%%%%%%%%%%%%%%%%%%%%%%%%%%%%%%%%%
\section{Greedy Sweep Rule}
%%%%%%%%%%%%%%%%%%%%%%%%%%%%%%%%%%%%%%%%%%%%
\label{app_sweep}

We can readily generalize the Sweep Rule in Definition~\ref{def_sweep2} to be applicable to any $d$-dimensional causal lattice $\mathcal{L}$, as long as $\mathcal{L}$ satisfies the assumptions discussed in Appendix~\ref{app_lattices}.
We remark that in $d$ dimensions there are $d-1$ different types of the Sweep Rule.
Namely, if we place a $\pm 1$ spin on every $k$-simplex of $\mathcal{L}$, where $k = 2,\ldots,d$, then the corresponding Sweep Rule of type $k$ governs the dynamics of the $(k-1)$-dimensional domain wall $\sigma\in \im\partial_{k}$ separating spins of different values.

As in the two-dimensional case, we first define a vertex $v$ to be trailing iff the restriction $\sigma|_v = \sigma \cap B_v(1)$ of the domain wall $\sigma$ to the neighborhood $B_v(1)$ of $v$ is non-empty and is contained in the future of $v$, i.e., $\sigma|_v \subseteq \star {k-1} v \cap \uset v$.
Then, for every trailing vertex $v$ the Sweep Rule finds a set of neighboring $k$-simplices $\varphi(v) \subseteq \star{k}{v} \cap \uset v$ in the future of $v$ satisfying two conditions:
\begin{itemize}
\item[(i)] the boundary of $\phi(v)$ locally matches the domain wall, i.e., $(\partial_{k} \phi(v))\rest v = \sigma \rest v$,
\item[(ii)] the causal diamonds of $\phi(v)$ and $\sigma\rest v$ match, i.e., $\cdia{\partial_{k} \phi(v)} = \cdia{\sigma \rest v}$,
\end{itemize}
then the Rule flips spins in $\phi(v)$.
We remark that the choice of $\phi(v)$ may not be unique unless $k=d-1$.
As a result, the domain wall is locally pushed away from the trailing vertex $v$ in the sweep direction $\vec t$.

We already mentioned in the main text that one could consider a CA without the condition on vertices to be trailing.
The resulting local spin update, the Greedy Sweep Rule, can be succinctly formulated as follows.
\begin{definition}[Greedy Sweep Rule in $d$ dimensions]
\label{def_greedy}
If a vertex $v$ belongs to the domain wall $\sigma$, then find a subset $\phi (v) \subseteq \star {k} v \cap \uset v$ of neighboring $k$-simplices in the future $\uset v$ satisfying the following three conditions
\begin{itemize}
\item[(i)] the boundary of $\phi(v)$ is locally contained in the domain wall, i.e., $(\partial_{k} \phi(v))\rest v \subseteq \sigma \rest v$,
\item[(ii)] the causal diamond of $\phi(v)$ is in the causal diamond of the domain wall restriction, i.e., $\cdia {\phi(v)} \subseteq \cdia{\sigma|_v}$,
\item[(ii)] the size of $\sigma\rest v + (\partial_k \phi)\rest v$ is minimal,
\end{itemize}
and flip spins on the $k$-simplices in $\phi(v)$.
\end{definition}
We emphasize that in Definition~\ref{def_greedy}  we try to capture a family of local rules corresponding to different possible ways of choosing $\varphi(v)$.
Also, the Sweep Rule in Definition~\ref{def_sweep2} is a special case of the Greedy Sweep Rule.
Namely, for trailing vertices the action of both of them turns out to be identical.
Moreover, their properties (as stated in Lemma~\ref{lemma_properties}) are the same and one can prove non-zero threshold for a decoder based on the Greedy Sweep Rule as well.

The Greedy Sweep Rule is local, since in order to apply it to some vertex $v$, we only require the knowledge of the restriction of the domain wall $\sigma|_v$ and the set of $k$-simplices $\star{k}{v}$ in the neighborhood of $v$.
The cardinality $|\star{k}{v} \cap \uset v|$ depends on the local structure of the lattice $\mathcal{L}$, but we are interested in cases when it is upper-bounded by some constant.
Thus, finding a subset $\varphi(v)$ can be done in constant time by checking all possible subsets of $\star{k}{v} \cap \uset v$ and finding the one satisfying conditions (i)--(iii).
In some special cases, e.g. in two dimensions, one can find $\varphi(v)$ more efficiently than via the exhaustive search; see~\cite{Kubicathesis}.
Note that for any trailing vertex $v$ and the domain wall $\sigma$ we can always find $\varphi(v)$ satisfying $(\partial_{k} \phi(v))\rest v = \sigma \rest v$ and $\cdia {\phi(v)} = \cdia{\sigma|_v}$ (this follows immediately from the assumption about the local structure of the lattice).

\begin{figure}[t!]
\centering
(a)\includegraphics[width=.4\columnwidth]{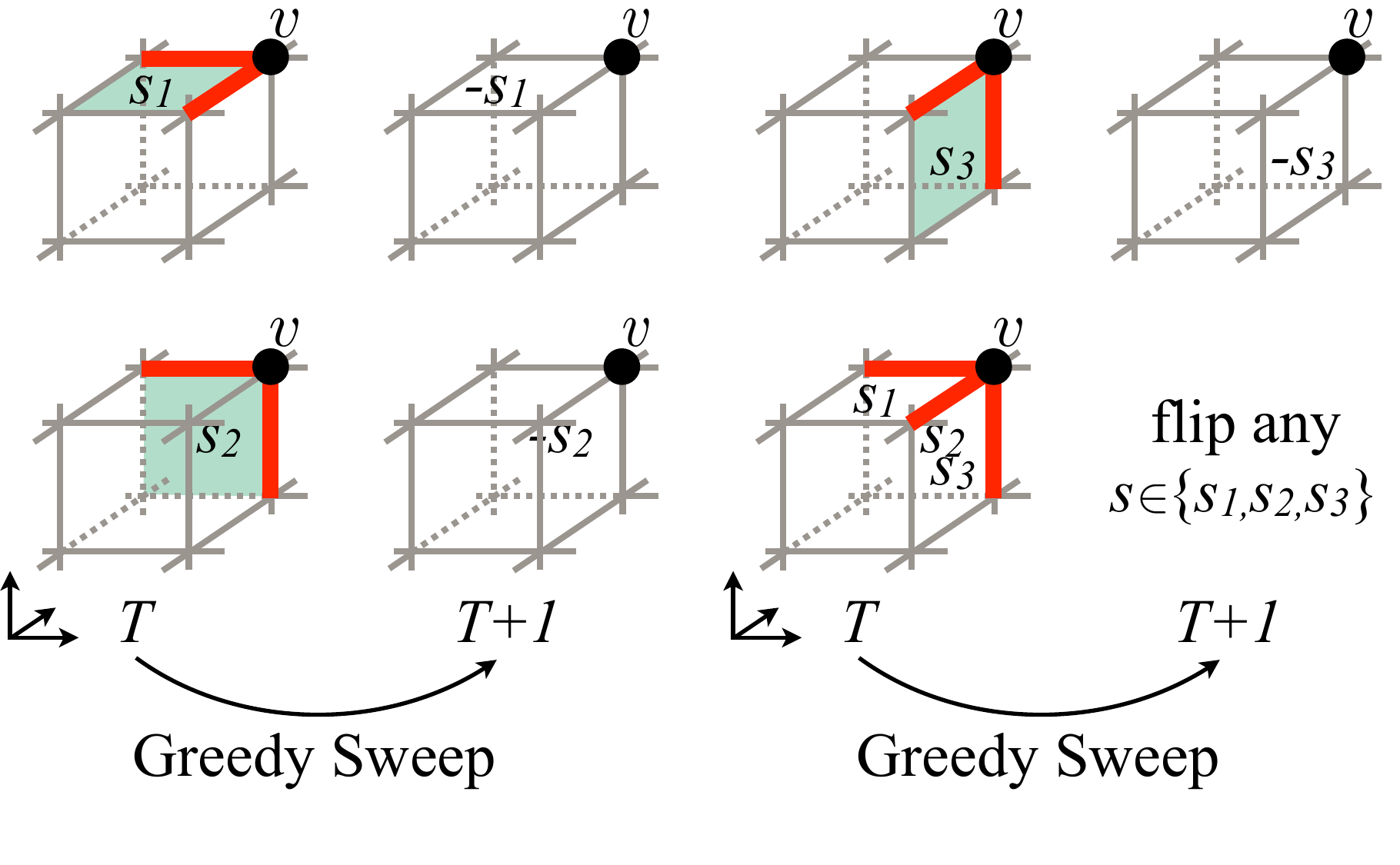}
\hfill(b)\includegraphics[width=.47\columnwidth]{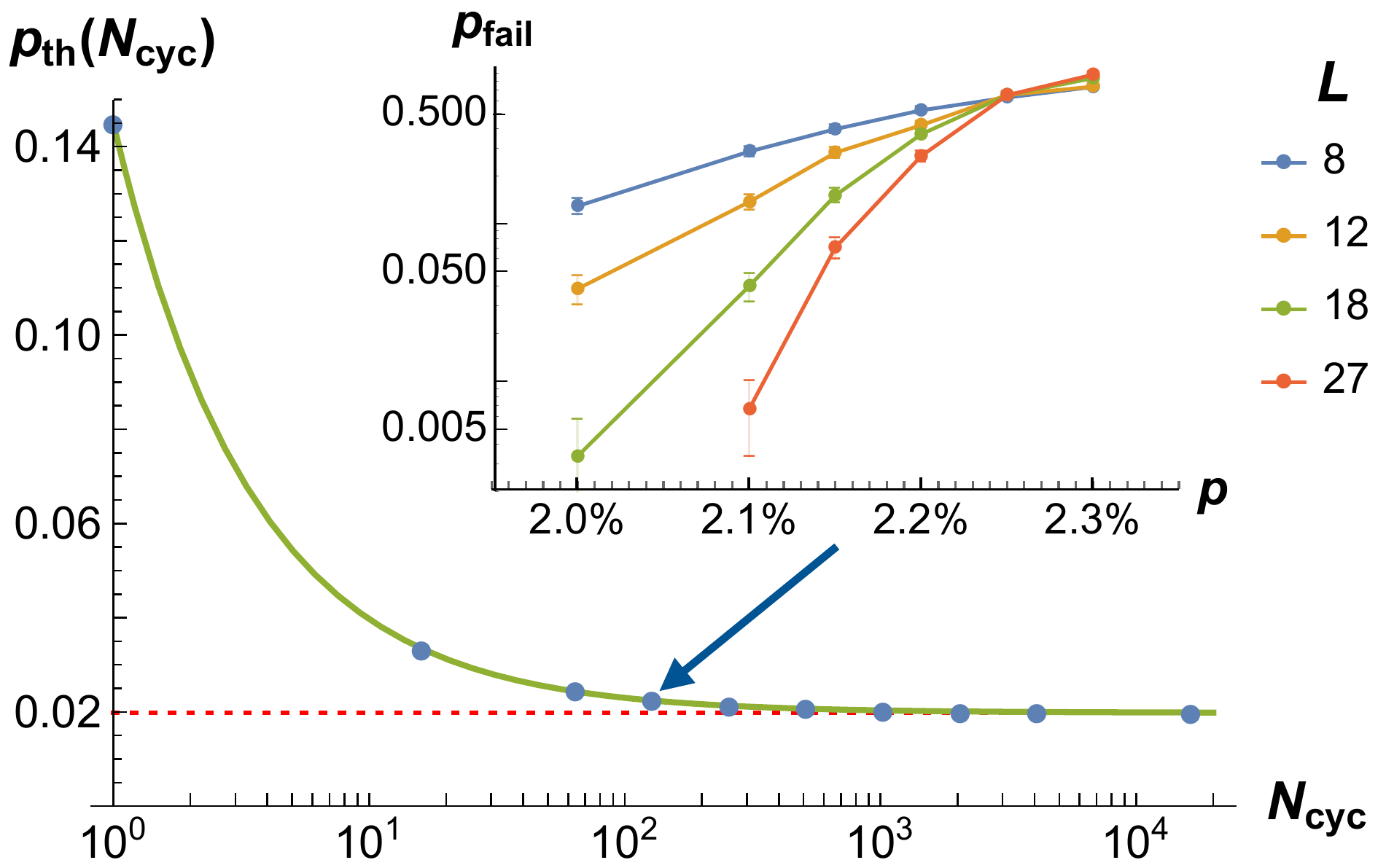}
\caption{
(a) A version of the Greedy Sweep Rule for the cubic lattice $\mathcal{L}$ with the sweep direction $\vec t = -(1,1,1)\in\mathbb{R}^3$ and classical $\pm 1$ spins on faces can be viewed as a higher-dimensional generalization of Toom's rule.
The local spin update rule is defined for every vertex $v$ of $\mathcal{L}$ and depends on the restriction $\sigma\rest v$ of the domain wall (red) to the edges incident to $v$, which are also in the future $\uset v$.
(b) (Main) The performance of a decoder based on the Greedy Sweep Rule for the 3D toric code on the cubic lattice $\mathcal{L}$ with qubits on faces and the phase-flip noise model.
Note that $p_{\textrm{th}}(1)$ corresponds to the threshold for perfect syndrome extraction.
As described in the main text, to estimate the sustainable threshold we fit the numerical ansatz from Eq.~(\ref{eq_ansatz_pheno}) (green line) to the data and find $p^{\textrm{cubic}}_{\textrm{sus}} = \psusvaltoom $ and $\gamma = \gammavaltoom$.
(Inset) The failure probability $p_{\textrm{fail}}(p,L)$ of the decoder after $N_{\textrm{cyc}} = 2^7$ correction cycles, where $p$ is the error rate and $L$ is the linear size of $\mathcal{L}$.
We consider a phenomenological noise model such that in each error correction cycle $Z$ errors on qubits occur with probability $p$, and in addition measured syndrome bits are flipped with probability $p$.
We estimate the threshold $p_{\textrm{th}}(N_{\textrm{cyc}}) \approx 2.25\%$ from the crossing point of different curves.
}
\label{fig_Toom3D}
\end{figure}

Lastly, we remark that one can construct variants of the Greedy Sweep Rule applicable to other $d$-dimensional locally Euclidean lattices not necessarily built of $d$-simplices but satisfing the causal structure properties (i) and (ii) from Appendix~\ref{app_lattices}.
Examples of such lattices include the square and cubic lattices, as well as the 2D parallelogram lattice in Fig.~\ref{fig_Toom}(b).
From that viewpoint, the Greedy Sweep Rule is a generalization of Toom's rule.
In particular, the Greedy Sweep Rule with the sweep direction $\vec{t} = -(1,1)\in\mathbb{R}^2$ would be identical to Toom's rule on the square lattice, and would not encounter any persistent domain wall configurations for the 2D parallelogram lattice.
In the case of the cubic lattice with the sweep direction $\vec{t} = -(1,1,1) \in \mathbb{R}^3$ and spins placed on faces, the Greedy Sweep Rule can be viewed as a higher-dimensional version of Toom's rule~\cite{Barabasi1992}; see Fig.~\ref{fig_Toom3D}(a) for an illustration of the local rules to update spins.
We numerical estimate the sustainable threshold $p^{\textrm{cubic}}_{\textrm{sus}} = \psusvaltoom $ of the decoder based on the Greedy Sweep Rule for the 3D toric code on the cubic lattice with qubits on faces and the phase-flip noise model; see Fig.~\ref{fig_Toom3D}(b).
One might speculate that a more general variant of the Greedy Sweep Rule can be defined for systems with causal structure that are not geometrically local.

%%%%%%%%%%%%%%%%%%%%%%%%%%%%%%%%%%%%%%%%%%%%
\section{Proof of properties of the Sweep Rule}
%%%%%%%%%%%%%%%%%%%%%%%%%%%%%%%%%%%%%%%%%%%%
\label{app_properties}

Before we present the proof of Lemma~\ref{lemma_properties}, we remark that the Support Property can be strengthen to read
\begin{equation}
\sigma^{(T)} \subset \left( \bigcup_{v \in \face 0 {\sigma}} \uset v \right) \cap \dset{\sup \sigma}.
\end{equation}
Also, we can make the bound in the Removal Property tighter, namely
\begin{equation}
T \geq \max_{v\in\face{0}{\sigma}} \max_{\cpth{v}{\sup \sigma}}  | \cpth{v}{\sup \sigma}|.
\end{equation}
However, in the proof of non-zero threshold it suffices to use weaker conditions, as stated in the main text in~Eqs.~(\ref{eq_support})~and~(\ref{eq_removal}), which are simpler to parse.

\begin{proof}

We prove the properties of the Sweep Rule by induction.
For $T=1$ all of them trivially hold.
In the rest of the proof, we use the following simple fact about causal diamonds:
\begin{itemize}
\item  for any finite $U,W \subset \face 0 {\mathcal{L}}$ if $U\subseteq W$, then $U\subseteq \cdia U \subseteq \cdia W$.
\end{itemize}

Now we show the induction step for the Support Property.
Let $V^{(T-1)}\subset \face 0 {\sigma^{(T-1)}}$ denote the set of trailing vertices of the domain wall $\sigma^{(T-1)}$ at time step $T-1$.
Note that in between time steps $T-1$ and $T$ the Sweep Rule finds for every trailing vertex $v\in V^{(T-1)}$ a certain subset $\phi^{(T-1)}(v)$ of neighboring $(k+1)$-simplices, which locally matches $\sigma^{(T-1)}$, i.e., $(\phi^{(T-1)}(v))\rest v = \sigma^{(T-1)}\rest v$.
Then, by flipping spins on $\phi^{(T-1)}(v)$ the domain wall is locally modified and it becomes
\begin{equation}
\sigma^{(T)} = \sigma^{(T-1)} + \sum_{v \in V^{(T-1)}} \phi^{(T-1)}(v).
\end{equation}
Note that $\phi^{(T-1)}(v)$ is chosen in such a way that $\cdia{\phi^{(T-1)}(v)} = \cdia{\sigma^{(T-1)}\rest v}$, and thus 
$\cdia{\phi^{(T-1)}(v)} \subseteq \cdia{\sigma^{(T-1)}} \subseteq \cdia \sigma$.
We conclude that 
\begin{equation}
\cdia{\sigma^{(T)}} \subseteq \cdia{\cdia{\sigma^{(T-1)}}\cup\bigcup_{v\in V^{(T-1)}} \cdia{\phi^{(T-1)}(v)}} \subseteq \cdia \sigma.
\end{equation}

It is straightforward to prove the Propagation Property.
Namely, every vertex $v$ in the domain wall $\sigma^{(T)}$ either belongs to $\sigma^{(T-1)}$ or is connected to some vertex $u\in\face 0 {\sigma^{(T-1)}}$ via an edge $(u,v) \in \face 1 {\mathcal{L}}$, such that $(u,v)\cdot \vec t > 0$.
Note that the latter case can arise when we locally modify $\sigma^{(T-1)}$ by flipping $(k+1)$-simplices around its trailing vertex $u$.
Thus, by using the induction hypothesis and triangle inequality we arrive at
\begin{equation}
\cdist{v,\sigma} \leq \cdist{v,u} + \cdist{u,\sigma} \leq T,
\end{equation}
where we set $u=v$ if $v\in\face 0 {\sigma^{(T-1)}}$.

To show the Time Property we {define} the integer-valued function
\begin{equation}
f_\sigma(T) = \max_{v\in\face 0 {\sigma^{(T)}}} \max_{\cpth{v}{\sup \sigma}} |\cpth{v}{\sup \sigma}|,
\end{equation}
which is the length of the longest causal path between the supremum of $\sigma$ and any vertex $v$ in the domain wall $\sigma^{(T)}$;
{ if $\sigma^{(T)} = 0$, then we set $f_\sigma(T) = 0$.}
{We argue that the function $f_\sigma(T)$ is a monotone of the Sweep Rule, namely it monotonically decreases with $T$ until the domain wall $\sigma^{(T)}$ is removed.}
First, note that if $v$ is a vertex of $\sigma^{(T)}$ which maximizes the function $f_\sigma(T)$, then it has to be trailing.
Thus, in between time steps $T$ and $T+1$ the Sweep Rule modifies the domain wall in the neighborhood of $v$.
In particular, $v$ is not included in $\sigma^{(T+1)}$, however some new vertices from the neighborhood of $v$, which are necessarily closer (in the sense of the longest causal path) to $\sup \sigma$ may be included.
Thus, we get $f_\sigma(T+1) < f_\sigma(T)$, as desired.

We observe that the Time Property follows immediately from the monotone $f_\sigma(T)$.
Namely, the initial value $f_\sigma(1)$ is upper-bounded by $\max_{\cpth{\inf\sigma}{\sup\sigma}}  |\cpth{\inf\sigma}{\sup\sigma}|$.
As long as the domain wall $\sigma^{(T)} \neq 0$, the monotone $f_\sigma(T)$ is decreased by at least one at each time step.
{Thus, for all $T > \max_{\cpth{\inf\sigma}{\sup\sigma}}  |\cpth{\inf\sigma}{\sup\sigma}|$ we necessarily have $f_\sigma(T) = 0$ and the domain wall is guaranteed to disappear, proving the Time Property. }
\end{proof}

%%%%%%%%%%%%%%%%%%%%%%%%%%%%%%%%%%%%%%%%%%%%
\section{Proof of threshold}
%%%%%%%%%%%%%%%%%%%%%%%%%%%%%%%%%%%%%%%%%%%%
\label{app_proof}

Now we are ready to prove the (Threshold) Theorem~\ref{thm_thres}.
Our proof is inspired by previous works \cite{Gacs1988, Harrington2004, Bravyi2013a} and consists of three parts.
In the first part, we discuss how to decompose the error configuration into recursively defined ``chunks,'' which naturally leads to the notion of a disjoint decomposition of the error configuration into connected components. 
Then, we explicitly find a positive constant $\pthr$ such that for phase-flip noise with rate $p < \pthr$ the probability of observing a level-$n$ chunk is doubly-exponentially suppressed in $n$.
Our notation for the chunk decomposition and the arguments about suppression of high-level chunks closely follow Ref.~\cite{Bravyi2013a}.
Finally, using the Sweep Rule Properties from Lemma~\ref{lemma_properties}, as well as the assumptions (i)--(v) on the family of considered lattices from Appendix~\ref{app_lattices} we show that the Sweep Decoder successfully corrects all connected components up to some level, which in turn allows us to upper-bound the decoding failure probability.

\subsection{Chunk decomposition and connected components}

Let $\epsilon\subseteq\face{k}{\mathcal{L}}$ be an error configuration in the $d$-dimensional toric code of type $k=2,\ldots,d-1$, i.e., the set of $k$-simplices identified with qubits affected by Pauli $Z$ errors.
We define a level-0 chunk $E^{[0]}$ to be an element of $\epsilon$.
In other words, a level-0 chunk corresponds to a single location of error.
We recursively define a level-$n$ chunk $E^{[n]} = E_1^{[n-1]} \sqcup E_2^{[n-1]}$ to be a disjoint union of two level-$(n-1)$ chunks $E_1^{[n-1]}$ and $E_2^{[n-1]}$, such that $\diam{E^{[n]}} \leq Q^n/2$
for some constant $Q$.
{Note that by a disjoint union $A\sqcup B$ we mean a union of two sets $A$ and $B$ which are disjoint, i.e., $A\cap B = \emptyset$.}
We define the level-$n$ error $E_n \subseteq \epsilon$ to be a union of all level-$n$ chunks
\begin{equation}
E_n = \bigcup_i E^{[n]}_i.
\end{equation}
Note that by definition $\epsilon = E_0$.
Also, we have the following sequence of inclusions
\begin{equation}
\epsilon = E_0 \supseteq E_1 \supseteq \ldots \supseteq E_m \supsetneq E_{m+1} = \emptyset,
\label{eq_error_inclusion}
\end{equation}
which allows us to define $F_i = E_i \setminus E_{i+1}$ for $i=0,1\ldots,m$.
Note that for any finite $\epsilon$ there exists a finite $m$ satisfying Eq.~(\ref{eq_error_inclusion}).
Lastly, we arrive at the following disjoint decomposition of the error configuration
\begin{equation}
\epsilon = F_0 \sqcup F_1 \sqcup \ldots \sqcup F_m.
\end{equation}

We say that a subset of errors $M\subseteq \epsilon$ is an $l$-connected component if it cannot be split into two disjoint non-empty sets $M_1$ and $M_2$ separated by more than $l$.
In other words, for any $M_1, M_2 \neq \emptyset$ if $M = M_1 \sqcup M_2$, then $d(M_1,M_2) \leq l$.
One can show that, roughly speaking, the diameter of any connected component is not too big and different connected components are far from each other.
This important observation is captured by the following Lemma~\ref{lemma_connected_comp}, whose proof we include for completeness.

\begin{lemma}[Connected Components~\cite{Bravyi2013a}]
\label{lemma_connected_comp}
Let $Q \geq 6$ be some constant and a subset of errors $M\subseteq \epsilon$ be a $Q^i$-connected component of $F_i$.
Then, $\diam{M} \leq Q^i$ and $d(M,E_i\setminus M) > Q^{i+1}/3$.
\end{lemma}

\begin{proof}
We prove the lemma by contradiction.
{For brevity, we say that a vertex $v$ is in some chunk $E$ if $v$ is a vertex of some $k$-simplex contained in $E$, i.e., $v\in \face 0 E$.}
Let us pick any $a\in \face{0}{F_i}$ and assume that there exists $b\in \face{0}{E_i }$, such that $Q^i < d(a,b) \leq Q^{i+1}/3$.
Then, $a$ and $b$ cannot be in the same level-$n$ chunk.
Hence $a$ and $b$ {are in} two different level-$n$ chunks $A$ and $B$, which are necessarily disjoint.
Using triangle inequality and $Q\geq 6$ we get
\begin{equation}
\diam{A\sqcup B} \leq \diam{A} + d(A,B) +\diam{B} \leq Q^i/2 + Q^{i+1}/3 + Q^i/2 \leq Q^{i+1}/2.
\end{equation}
This implies that $A\sqcup B$ is a level-$(i+1)$ chunk, and therefore $a$ is in $A\sqcup B \subseteq E_{i+1}$, which is in contradiction with {$a$ being in $F_i = E_{i}\setminus E_{i+1}$}.
We thus conclude that for any $b\in \face 0 {E_i}$ we either have $d(a,b)< Q^i$ or $d(a,b) > Q^{i+1}/3$.
The former case leads us to a conclusion that any $Q^i$-connected component $M\subseteq F_i$ has diameter at most $Q^i$.
The latter case allows us to argue that the distance between $M$ and $E_i\setminus M$ is more than $Q^{i+1}/3$.
\end{proof}

We remark that the (Connected Components) Lemma~\ref{lemma_connected_comp} will be used to show that the Sweep Decoder removes different connected components independently of one another since they are sufficiently far apart.

\subsection{Suppression of high-level chunks}

Let us consider a discrete $d$-dimensional ball $B_v(Q^n/2)$ of radius $Q^n/2$, where $Q$ is some constant, centered at the vertex $v$, and also consider its cover with smaller balls of radius $Q^{n-1}/2$ indexed by $U \subseteq \face 0 {\mathcal{L}}$, i.e., $B_v(Q^n/2) = \bigcup_{u\in U} B_u(Q^{n-1}/2)$.
We assume that the qubits are independently affected by $Z$ error with probability $p$.
Let $\epsilon\subseteq \face k {\mathcal{L}}$ be a randomly chosen error configuration.
We will argue that the probability of the ball $B_v(Q^n/2)$ intersecting a level-$n$ chunk (which itself is composed of two disjoint level-$(n-1)$ chunks) of $\epsilon$ is upper-bounded by some function of the probabilities of the balls $B_u(Q^{n-1}/2)$ intersecting a level-$(n-1)$ chunk of $\epsilon$ for all $u\in U$.
Namely, we will prove the following
\begin{equation}
\pr{\textrm{$B_v(Q^n/2)$ intersects a level-$n$ chunk of $\epsilon$}} \leq
\left(\sum_{u\in U} \pr{\textrm{$B_u(Q^{n-1}/2)$ intersects a level-$(n-1)$ chunk of $\epsilon$}}\right)^2\hspace*{-5pt}.
\label{eq_chunk_iteration}
\end{equation}
Then, the above inequality combined with the property of the lattice $\mathcal{L}$ being locally Euclidean will allow us to show that
if the single-qubit $Z$ error probability $p$ is below $\pthr$ defined in Eq.~(\ref{eq_pth}), then the probability of having a level-$n$ chunk in a randomly chosen error configuration $\epsilon\subseteq\face k {\mathcal{L}}$ is doubly-exponentially small in $n$. Namely,
\begin{equation}
\pr{\textrm{a level-$n$ chunk in $\epsilon$}} \leq |\face 0 {\mathcal{L}}| \lambda^{-2} \left(\frac{p}{\pthr}\right)^{2^n}, 
\label{eq_chunk_suppression}
\end{equation}
where $\lambda = (2Q)^d c_B$, $d$ is the dimension of $\mathcal{L}$ and $c_B$ is a constant defined for $\mathcal{L}$ via Eq.~(\ref{eq_balls}).

We start by defining the sample space $\Omega = 2^{\face k {\mathcal{L}}}$ as the collection of all possible $Z$ error configurations $\epsilon\subseteq \face k {\mathcal{L}}$, where $\face k {\mathcal{L}}$ is countable.
Let an event $\mathcal{E} = \{ \epsilon_1,\epsilon_2,\ldots \} \subseteq \Omega$ be a collection of some error configurations $\epsilon_i \subseteq \face k {\mathcal{L}}$.
We say that the event $\mathcal{E}$ is increasing if $\epsilon\in\mathcal{E}$ implies $\epsilon'\in\mathcal{E}$ for any two configurations $\epsilon\subseteq \epsilon' \subseteq \face k {\mathcal{L}}$.
The disjoint occurrence $\mathcal{E}\circ \mathcal{E}'$ of two events $\mathcal{E}$ and $\mathcal{E}'$ is defined as the collection of configurations $\epsilon \sqcup \epsilon'$, which are the disjoint union of $\epsilon\in\mathcal{E}$ and $\epsilon'\in\mathcal{E}'$, i.e., $\epsilon \cap \epsilon' = \emptyset$.

We can introduce a probability measure $\textrm{pr}:2^\Omega \rightarrow [0,1]$, which assigns the probability $\pr{\mathcal{E}}$ to any event $\mathcal{E}\subseteq {\Omega}$.
We assume that the probability measure satisfies the following condition: for all $k$-simplices $\delta\in\face k {\mathcal{L}}$ the events $\{ \epsilon \subseteq \face k {\mathcal{L}} | \delta\in\epsilon\}$ are independent under the probability measure and
\begin{equation}
\pr{\{ \epsilon \subseteq \face k {\mathcal{L}} | \delta\in\epsilon\}} = 1 - \pr{\{ \epsilon \subseteq \face k {\mathcal{L}} | \delta\not\in\epsilon\}}.
\end{equation}
In other words, we assume that each qubit associated with $\delta\in\face k {\mathcal{L}}$ is independently affected by Pauli $Z$ error with probability $p_\delta$.
For simplicity, we further assume that the error probability $p_\delta$ is the same for every qubit and equal to $p$.
We remark that in the case of the lattice $\mathcal{L}$ with a finite number of $k$-simplices, the probability of the error configuration $\epsilon$ is given by $\pr{\epsilon} = p^{|\epsilon|}(1-p)^{|\face k {\mathcal{L}}| - |\epsilon|}$.
Now we are ready to state the van den Berg and Kesten inequality~\cite{VanDenBerg1985}, which is central to our proof: if $\mathcal{E}$ and $\mathcal{E}'$ are two increasing events, then the probability $\pr{\mathcal{E}\circ\mathcal{E}'}$ of the disjoint occurrence of $\mathcal{E}$ and $\mathcal{E}'$ is upper-bounded by $\pr{\mathcal{E}}\pr{\mathcal{E}'}$.

To prove Eq.~(\ref{eq_chunk_iteration}) and the suppression of high-level chunks in Eq.~(\ref{eq_chunk_suppression}) we define the following increasing events
\begin{itemize}
\item $\mathcal{A}_{v,n} = 
\{ \epsilon\subseteq\face k {\mathcal{L}}|\textrm{$B_v(Q^n/2)$ intersects a level-$n$ chunk of $\epsilon$}\}$,
\item $\mathcal{B}_{v,n} = 
\{ \epsilon\subseteq\face k {\mathcal{L}}|\textrm{$B_v(Q^n)$ contains a level-$n$ chunk of $\epsilon$}\}$,
\item $\mathcal{C}_{v,n} = 
\{ \epsilon\subseteq\face k {\mathcal{L}}|\textrm{$B_v(Q^n)$ contains two disjoint level-$(n-1$ chunks of $\epsilon$}\}$,
\item $\mathcal{D}_{v,n} = 
\{ \epsilon\subseteq\face k {\mathcal{L}}|\textrm{$B_v(Q^n)$ contains a level-$(n-1)$ chunk of $\epsilon$}\}$.
\end{itemize}
In words, an increasing event $\mathcal{A}_{v,n}$ is defined as the set of all error configurations such that the ball $B_v(Q^n/2)$ has a non-zero overlap with a level-$n$ chunk of each of those configurations; similarly $\mathcal{B}_{v,n}$, $\mathcal{C}_{v,n}$ and $\mathcal{D}_{v,n}$.
By definition of chunks we have
\begin{equation}
\pr{\mathcal{A}_{v,n}} \leq \pr{\mathcal{B}_{v,n}} \leq \pr{\mathcal{C}_{v,n}}.
\label{eq_event_chain}
\end{equation}
To relate the probabilities of events $\mathcal{C}_{v,n}$ and $\mathcal{D}_{v,n}$ we first note that the event $\mathcal{C}_{v,n}$ is the disjoint occurrence of $\mathcal{D}_{v,n}$ and $\mathcal{D}_{v,n}$, i.e., $\mathcal{C}_{v,n} = \mathcal{D}_{v,n} \circ \mathcal{D}_{v,n}$.
Then, using the van den Berg and Kesten inequality we find $\pr{\mathcal{C}_{v,n}} \leq  \pr{\mathcal{D}_{v,n}}^2$. 
Now, consider a cover of the ball $B_v(Q^n)$ with balls of radius $Q^{n-1}/2$ indexed by $U\subseteq\face 0 {\mathcal{L}}$.
If the event $\mathcal{D}_{v,n}$ happens, i.e., the ball $B_v(Q^n)$ contains a level-$(n-1)$ chunk, then there exists a vertex $u\in U$, such that the ball $B_u (Q^{n-1}/2)$ has non-zero overlap with that chunk.
Notice that the latter condition describes the event $\mathcal{A}_{u,n-1}$.
Thus, using the union bound we arrive at
\begin{equation}
\pr{\mathcal{D}_{v,n}} \leq \sum_{u\in U} \pr{\mathcal{A}_{u,n-1}},
\end{equation}
which combined with Eq.~(\ref{eq_event_chain}) results in Eq.~(\ref{eq_chunk_iteration}).
Now we invoke the property of the lattice $\mathcal{L}$ being locally Euclidean.
This property guarantees that we can find a cover with $|U| \leq \lambda$, where $\lambda = (2Q)^d c_B$, $d$ is the dimension of $\mathcal{L}$ and $c_B$ is the constant defined via Eq.~(\ref{eq_balls}).
Thus, we obtain
\begin{equation}
\pr{\mathcal{A}_{v,n}} \leq \left(\sum_{u\in U} \pr{\mathcal{A}_{u,n-1}} \right)^2 \leq 
\left(|U| \max_{u\in U} \pr{\mathcal{A}_{u,n-1}} \right)^2 \leq  \left(\lambda \max_{u\in U} \pr{\mathcal{A}_{u,n-1}}\right)^2.
\label{eq_avn}
\end{equation}
Let us denote the probability of the event $\mathcal{A}_{v,n}$ maximized over the set of vertices $\face 0 {\mathcal{L}}$ by
\begin{equation}
p_{\mathcal{A},n} = \max_{v\in\face{0}{\mathcal{L}}} \pr{\mathcal{A}_{v,n}}
\end{equation}
By recursively using Eq.~(\ref{eq_avn}) we can conclude that
\begin{equation}
p_{\mathcal{A},n} \leq (\lambda p_{\mathcal{A},n-1})^2 \leq \ldots \leq \lambda^{-2} \left(\lambda^2 p_{\mathcal{A},0}\right)^{2^n}
\end{equation}
and therefore $p_{\mathcal{A},n}$ is doubly-exponentially small in $n$ for $p_{\mathcal{A},0} < \lambda^{-2}$. 
Note that the event ${\mathcal{A}_{w,0}}$ describes the situation that at least one qubit in the neighborhood of the vertex $w$ is affected by the error configuration $\epsilon$.
Thus, $\pr{\mathcal{A}_{w,0}}$ is upper bounded by $|\star k w| p$ and subsequently $p_{\mathcal{A},0} \leq \max_{w\in\face 0 {\mathcal{L}}} |\star k w| p$.
We observe that if the error probability $p$ is below $\pthr$ defined as
\begin{equation}
\pthr = \left(((2Q)^d c_B)^2 \max_{v\in \face 0 {\mathcal{L}}} |\star k v|\right)^{-1},
\label{eq_pth}
\end{equation}
then $p_{\mathcal{A},0} < \lambda^{-2}$.
Finally, we note that if a randomly chosen error configuration $\epsilon$ contains a level-$n$ chunk, then for some $v\in\face 0 {\mathcal{L}}$ 
the ball $B_v(Q^n/2)$ has to intersect that chunk. Thus, using union bound we get
\begin{equation}
\pr{\textrm{a level-$n$ chunk in $\epsilon$}} \leq \sum_{v\in\face 0 {\mathcal{L}}} \pr{\mathcal{A}_{v,n}}
\leq |\face 0 {\mathcal{L}}| p_{\mathcal{A},n},
\end{equation}
which leads to Eq.~(\ref{eq_chunk_suppression}).
We remark that in the following subsection we will see that $\pthr$ serves as the lower-bound on the Sweep Decoder threshold if we choose  $Q = 6 c_D c_P$ with $c_D$ and $c_P$ defined in the discussion of the properties of causal lattices via Eqs.~(\ref{eq_diameters})~and~(\ref{eq_paths}), respectively.

\subsection{Putting things together}
\label{app_proof_together}

Now we are ready to prove that the Sweep Decoder for the $d$-dimensional toric code of type $k = 2,\ldots,d-1$ has non-zero threshold, which is lower-bounded by $\pthr$ defined in Eq.~(\ref{eq_pth}).
For concreteness, we consider a family of lattices $\mathcal{L}$ on the $d$-dimensional torus of growing linear size $L\rightarrow \infty$, which satisfy conditions (i)--(v) from Appendix~\ref{app_lattices} {and $|\face 0 {\mathcal{L}}| = \textrm{poly}(L)$.}
Note that by the linear size of $\mathcal{L}$ we mean the length of the shortest non-contractible path in $\mathcal{L}$.
We remark that the toric code of type $k$ defined on the $d$-dimensional {torus} has $d \choose k$ logical qubits and the corresponding logical $Z$ operators can be represented as Pauli $Z$ operators with support forming non-contractible $k$-dimensional surfaces.

Recall that we consider the phase-flip noise model, i.e., each qubit is independently affected by a $Z$ error with probability $p$.
Let $\epsilon\subseteq \face k {\mathcal{L}}$ be a randomly chosen error {configuration, i.e., the set of $k$-simplices identified with qubits affected by $Z$ errors.}
The main idea behind the proof is to show that {the Sweep Decoder can successfully correct any level-$n$ chunk of the error configuration $\epsilon$ for all $n < m^* = \lceil\log_Q (L/c_D)\rceil$.
This will imply that the Sweep Decoder can fail only if there exists a level-$m^*$ chunk of the error configuration $\epsilon$.
Using Eq.~(\ref{eq_chunk_suppression}) we arrive at an upper-bound on the decoding failure probability 
\begin{equation}
\pr{\textrm{fail}} \leq \pr{\textrm{level-$m^*$ chunk}} \leq |\face 0 {\mathcal{L}}|\lambda^{-2} \left( \frac{p}{\pthr}\right)^{2^{m^*}}
\leq \lambda^{-2} \textrm{poly}(L) \left( \frac{p}{\pthr}\right)^{\alpha L^\beta},
\label{eq_final}
\end{equation}
which goes to zero in the limit of infinite linear size $L\rightarrow \infty$ for $p < \pthr$.
Here, $\lambda = (2 Q)^d c_B$, $\alpha = c_D^{-\beta}$, $\beta = \log_Q 2$, $Q=6 c_D c_P$, $\pthr$ is a positive constant specified in Eq.~(\ref{eq_pth}), $d$ is the dimension of $\mathcal{L}$, and $c_B$, $c_D$, $c_P$ are the constants defined via Eqs.~(\ref{eq_balls}),~(\ref{eq_diameters}),~(\ref{eq_paths}), respectively.
Finally, we conclude that the threshold of the Sweep Decoder for the $d$-dimensional toric code of type $k = 2, \ldots, d-1$ is lower-bounded by $\pthr$.}

The last piece of the proof is to justify that for all $n<m^*$ any level-$n$ chunk can be successfully corrected by the Sweep Decoder.
First, {let $\epsilon = F_0 \sqcup F_1 \sqcup \ldots $ be the disjoint decomposition of the error configuration $\epsilon$ and} choose a constant $Q = 6 c_D c_P$
At every time step $T=1,2,\ldots$ the Sweep Decoder simultaneously applies the Sweep Rule to every vertex of the lattice and locally modifies the domain wall $\sigma^{(T)}\in \im \partial_k$, where we set $\sigma^{(1)} = \partial_k \epsilon$.
Consider any non-empty subset of errors $M\subseteq \epsilon$, which is a $Q^0$-connected component of $F_0$.
Then, within the first $T_0 = c_D c_P Q^0$ time steps the Sweep Rule removes the part $\partial_k M$ of the domain wall $\partial_k \epsilon$, which corresponds to $M$.
Namely, using the (Connected Components) Lemma~\ref{lemma_connected_comp} we get that $\diam M \leq Q^0$.
\footnote{Note that $\diam M = 1$ implies that all the vertices of $M$ belong to the same $d$-simplex $\delta$.
This, however, does not imply that the Sweep Rule can remove the corresponding part $\partial_k M$ of the domain wall in one step.
Rather, at most $d-1$ time steps may be required, as can be seen in the case of the one-dimensional domain wall visiting all vertices of $\delta$ in a sequence induced by the sweep direction $\vec t$.}
Since $\mathcal{L}$ is locally Euclidean, from Eq.~(\ref{eq_diameters}) we get
$c_D \diam M \geq \diam{\cdia M} \geq \dist{\inf M}{\sup M}$,
which combined with Eq.~(\ref{eq_paths}) results in the bound
$|\cpth{\inf M}{\sup M}| \leq c_D c_P \diam M = T_0$
on the length of any causal path within the causal diamond $\cdia M$.
Note that $\cdia{\partial_k M} \subseteq \cdia M$, and thus from the Removal Property in Lemma~\ref{lemma_properties} we obtain that $\partial_k M$ is guaranteed to be removed by time $T_0$, since $T_0 \geq \max_{\cpth{\inf\sigma^{(1)}}{\sup\sigma^{(1)}}}  |\cpth{\inf\sigma^{(1)}}{\sup\sigma^{(1)}}|$.

Importantly, in this reasoning we use the fact that the distance between $\partial_k M$ and $\partial_k\epsilon \setminus \partial_k M$ is greater than $Q^1/3$.
This fact follows from the (Connected Components) Lemma~\ref{lemma_connected_comp}.
Thus, the time evolution of the rest of the domain wall $\partial_k \epsilon \setminus \partial_k M$ due to the Sweep Rule does not affect the removal of $\partial_k M$.
This follows from the fact that both $\partial_k \epsilon \setminus \partial_k M$ and $\partial_k M$ can only propagate over the distance at most $T_0 = c_D c_P Q^0 \leq Q^1/6$ toward each other; see the Propagation Property in Lemma~\ref{lemma_properties}.
Thus, they will not cover the total distance of more than $Q^1/3$, which is the separation between them.

We remark that the reasoning is applicable to $Q^i$-connected components of $F_i$ for higher levels $i\geq 1$.
We summarize our discussion in the following lemma, which can be analogously proven by induction on the level $i$.
\begin{lemma}
\label{lemma_sweepcorrect}
Let $\epsilon\subseteq\face{k}{\mathcal{L}}$ be an error configuration with the disjoint decomposition
$\epsilon = F_0 \sqcup F_1 \sqcup \ldots$ and choose $Q = 6c_D c_P$. 
Then, for any $Q^i$-connected component $M$ of $F_i$ the corresponding part $\partial_k M$ of the domain wall $\partial_k \epsilon$ is removed by the Sweep Rule within first $T_i = c_D c_P Q^i$ time steps.
Moreover, the removal of $\partial_k M$ is not affected by any other part $\partial_k M'$ of the domain wall, irrespective of the level $j$ of the $Q^j$-connected component $M'$ of $F_j$.
\end{lemma}

We run the Sweep Decoder for $\tmax = T_{m^*-1} + 1 = O(L)$ time steps.
Then, Lemma~\ref{lemma_sweepcorrect} guarantees that by time $\tmax-1$ any $Q^n$-connected component $M$ of $F_n$ is removed for all $n<m^*$.
Moreover, for each $M$ the Sweep Decoder finds (independently of the other connected components) a correction of the part $\partial_k M$ of the domain wall, which is contained in the causal diamond of $\cdia M$.
This follows from the Support Property in Lemma~\ref{lemma_properties}.
Note that the diameter of the causal diamond $\cdia M$ is smaller than the linear size of the system
\begin{equation}
\diam{\cdia M} \leq c_D\cdot \diam M \leq c_D Q^n < c_D Q^{m^*} \leq L,
\end{equation}
where we use Eq.~(\ref{eq_diameters}) and the (Connected Components) Lemma~\ref{lemma_connected_comp}.
Thus, any operator supported within $\cdia M$ cannot implement a non-trivial logical operator.
This finishes the argument that the Sweep Decoder successfully corrects any level-$n$ chunk for $n<m^*$ by time $\tmax$.

We emphasize that we did not optimize the proof to maximize the threshold lower bound $\pthr$ in Eq.~(\ref{eq_pth}).
To illustrate the discrepancy between the bound and the actual threshold value, let us consider the 3D toric code on the bcc lattice, whose parameters are
$d=3$, $c_B = 24$, $c_D = 2$, $c_P = 3$ and $|\star 2 v | = 36$.
Then, from Eq.~(\ref{eq_pth}) we obtain a lower bound on the Sweep Decoder threshold to be $\pthr \approx 10^{-15}$, whereas the numerically estimated threshold is $p_{\textrm{th}}(1) \approx .0785$; see Fig.~\ref{fig_numerics}.
This example illustrates the importance of numerical estimates of threshold values.

\bibliography{bib_toom_pruned,bib_toom_extra}

\end{document}